\newif\ifdoublecol

\doublecoltrue
\documentclass[twocolumn]{IEEEtran}
\IEEEoverridecommandlockouts

\usepackage{etex} 

\ifCLASSINFOpdf
  \usepackage[pdftex]{graphicx}

  \usepackage{epstopdf}
\else
  \usepackage[dvips]{graphicx}
  \graphicspath{{../eps/}}
  \DeclareGraphicsExtensions{.eps}
\fi
\usepackage{cite}

\usepackage[cmex10]{amsmath}
\usepackage{xcolor,colortbl}
\definecolor{col1}{HTML}{3891A6}
\definecolor{col2}{HTML}{EF5B5B}
\definecolor{col3}{HTML}{3DDC97}

\usepackage{amsmath,amsthm,relsize}
\usepackage{amsfonts, amssymb, cuted}
\usepackage{hyperref}      

\hypersetup{
    colorlinks=true,       
    linkcolor=blue,        
    citecolor=blue,        
    filecolor=blue,        
    urlcolor=blue          
}
\usepackage{cleveref}
\usepackage{mathtools}
\usepackage{algorithm}
\usepackage[noend]{algpseudocode}
\usepackage{cite}
\usepackage{soul}
\usepackage{tikz}
\usepackage{pgfplotstable}
\usepackage{pgfplots}
\usepackage{nicefrac}
\usepackage{enumitem}
\usepackage{support-caption}
\usepackage{subcaption}
\usepackage[font=footnotesize,labelfont=bf]{caption}
\usepackage{multirow}
\pgfplotsset{compat=1.15}
\usepackage{relsize}
\usetikzlibrary{patterns}
\usepackage{acro}
\usepackage{pifont}

\usepackage{todonotes}

\usepackage{mathtools}

\usepackage{array}
\usepackage{booktabs}

\usepackage{diagbox}
\usepackage{smartdiagram}
\usesmartdiagramlibrary{additions}
\usepackage{bm}
\usepackage{multicol}
\usepackage{subcaption}
\renewcommand{\thesubfigure}{\arabic{figure}} 
\renewcommand{\figurename}{Fig.}

\usepackage{tabularx}
\usepackage{colortbl}

\usetikzlibrary{plotmarks}
  \usetikzlibrary{arrows.meta}
  \usepgfplotslibrary{patchplots}
  \usepackage{grffile}
  \pgfplotsset{plot coordinates/math parser=false}
  \newlength\figureheight
  \newlength\figurewidth
   \pgfplotsset{compat=1.11,
    /pgfplots/ybar legend/.style={
    /pgfplots/legend image code/.code={%
       \draw[##1,/tikz/.cd,yshift=-0.25em]
        (0cm,0cm) rectangle (3em,8pt);},
   },
}

\usepgfplotslibrary{groupplots,units}
\pgfplotsset{
  compat=1.9,
  unit code/.code 2 args={\si{#1#2}} 
}
\usepackage{siunitx}

\usepackage{color}
\usepackage{colortbl}
\usepackage{multirow}

\newlength{\Oldarrayrulewidth}

\definecolor{intnull}{RGB}{213,229,255}
\definecolor{inteins}{RGB}{128,179,255}
\definecolor{intzwei}{RGB}{42,127,255}
\definecolor{intdrei}{RGB}{0,85,212}
\definecolor{intvier}{RGB}{0,51,128}
\definecolor{intfunf}{RGB}{0,34,85}

\usetikzlibrary{decorations.pathreplacing}

\newtheorem{lemma}{Lemma}

\usepackage{arydshln}

\newcommand{\herm}{^{\mbox{\scriptsize H}}}

\newcommand{\vbar}{\raisebox{.17ex}{\rule{.04em}{1.35ex}}}
\newcommand{\vbarind}{\raisebox{.01ex}{\rule{.04em}{1.1ex}}}
\newcommand{\R}{\ifmmode{\rm I}\hspace{-.2em}{\rm R} \else ${\rm I}\hspace{-.2em}{\rm R}$ \fi}
\newcommand{\T}{\ifmmode{\rm I}\hspace{-.2em}{\rm T} \else ${\rm I}\hspace{-.2em}{\rm T}$ \fi}
\newcommand{\N}{\ifmmode{\rm I}\hspace{-.2em}{\rm N} \else \mbox{${\rm I}\hspace{-.2em}{\rm N}$} \fi}
\newcommand{\B}{\ifmmode{\rm I}\hspace{-.2em}{\rm B} \else \mbox{${\rm I}\hspace{-.2em}{\rm B}$} \fi}
\newcommand{\Hil}{\ifmmode{\rm I}\hspace{-.2em}{\rm H} \else \mbox{${\rm I}\hspace{-.2em}{\rm H}$} \fi}
\newcommand{\C}{\ifmmode\hspace{.2em}\vbar\hspace{-.31em}{\rm C} \else \mbox{$\hspace{.2em}\vbar\hspace{-.31em}{\rm C}$} \fi}
\newcommand{\Cind}{\ifmmode\hspace{.2em}\vbarind\hspace{-.25em}{\rm C} \else \mbox{$\hspace{.2em}\vbarind\hspace{-.25em}{\rm C}$} \fi}
\newcommand{\Q}{\ifmmode\hspace{.2em}\vbar\hspace{-.31em}{\rm Q} \else \mbox{$\hspace{.2em}\vbar\hspace{-.31em}{\rm Q}$} \fi}
\newcommand{\Z}{\ifmmode{\rm Z}\hspace{-.28em}{\rm Z} \else ${\rm Z}\hspace{-.28em}{\rm Z}$ \fi}

\newtheorem{lem}{Lemma}

\newtheorem{exmp}{Example}
\theoremstyle{definition}

\newcommand{\CB}[0]{{\mathcal{B}}}

\newcommand{\CF}[0]{{\mathcal{F}}}

\newcommand{\CK}[0]{{\mathcal{K}}}

\newcommand{\CP}[0]{{\mathcal{P}}}

\newcommand{\CT}[0]{{\mathcal{T}}}

\newcommand{\Bw}[0]{{\mathbf{w}}}
\newcommand{\Bx}[0]{{\mathbf{x}}}
\newcommand{\By}[0]{{\mathbf{y}}}
\newcommand{\Bz}[0]{{\mathbf{z}}}

\newcommand{\BH}[0]{{\mathbf{H}}}

\newcommand{\BU}[0]{{\mathbf{U}}}

\newcommand{\SfA}[0]{{\mathsf{A}}}
\newcommand{\SfB}[0]{{\mathsf{B}}}
\newcommand{\SfC}[0]{{\mathsf{C}}}

\newcommand{\SfE}[0]{{\mathsf{E}}}
\newcommand{\SfF}[0]{{\mathsf{F}}}

\newcommand{\SfP}[0]{{\mathsf{P}}}

\newcommand{\SfS}[0]{{\mathsf{S}}}

\newcommand{\SfU}[0]{{\mathsf{U}}}
\newcommand{\SfV}[0]{{\mathsf{V}}}

\usepackage{tikz}

\usepackage{tikz}
\usetikzlibrary{arrows,
                calc,
                decorations.pathreplacing,
                calligraphy,
                matrix,
                positioning
                }

\DeclareAcronym{ADMM}{
    short = ADMM,
    long = alternating direction method of multipliers,
    list = Alternating Direction Method of Multipliers,
    tag = abbrev
}

\DeclareAcronym{AoA}{
    short = AoA,
    long = angle-of-arrival,
    list = Angle-of-Arrival,
    tag = abbrev
}

\DeclareAcronym{SISO}{
    short = SISO,
    long = single-input single-output,
    list = single-input single-output,
    tag = abbrev
}

\DeclareAcronym{MRT}{
    short = MRT,
    long = maximum ratio transmitter,
    list = maximum ratio transmitter,
    tag = abbrev
}

\DeclareAcronym{PDA}{
    short = PDA,
    long = placement delivery array,
    list = placement delivery array,
    tag = abbrev
}

\DeclareAcronym{EE}{
    short = EE,
    long = energy efficiency,
    list = energy efficiency,
    tag = abbrev
}

\DeclareAcronym{MDS}{
    short = MDS,
    long = maximum distance separation,
    list = maximum distance separation,
    tag = abbrev
}

\DeclareAcronym{SIC}{
    short = SIC,
    long = successive-interference-cancellation,
    list = successive-interference-cancellation,
    tag = abbrev
}

\DeclareAcronym{MAC}{
    short = MAC,
    long = multiple-access-channel,
    list = multiple-access-channel,
    tag = abbrev
}

\DeclareAcronym{AoD}{
    short = AoD,
    long = angle-of-departure,
    list = Angle-of-Departure,
    tag = abbrev
}

\DeclareAcronym{BB}{
    short = BB,
    long = base band,
    list = Base Band,
    tag = abbrev
}

\DeclareAcronym{BC}{
    short = BC,
    long = broadcast channel,
    list = Broadcast Channel,
    tag = abbrev
}

\DeclareAcronym{BS}{
    short = BS,
    long = base station,
    list = Base Station,
    tag = abbrev
}

\DeclareAcronym{BR}{
    short = BR,
    long = best response,
    list = Best Response, 
    tag = abbrev
}

\DeclareAcronym{CB}{
    short = CB,
    long = coordinated beamforming,
    list = Coordinated Beamforming,
    tag = abbrev
}

\DeclareAcronym{CC}{
    short = CC,
    long = coded caching,
    list = Coded Caching,
    tag = abbrev
}

\DeclareAcronym{CE}{
    short = CE,
    long = channel estimation,
    list = Channel Estimation,
    tag = abbrev
}

\DeclareAcronym{CoMP}{
    short = CoMP,
    long = coordinated multi-point transmission,
    list = Coordinated Multi-Point Transmission,
    tag = abbrev
}

\DeclareAcronym{CRAN}{
    short = C-RAN,
    long = cloud radio access network,
    list = Cloud Radio Access Network,
    tag = abbrev
}

\DeclareAcronym{CSE}{
    short = CSE,
    long = channel specific estimation,
    list = Channel Specific Estimation,
    tag = abbrev
}

\DeclareAcronym{CSI}{
    short = CSI,
    long = channel state information,
    list = Channel State Information,
    tag = abbrev
}

\DeclareAcronym{CSIT}{
    short = CSIT,
    long = channel state information at the transmitter,
    list = Channel State Information at the Transmitter,
    tag = abbrev
}

\DeclareAcronym{CU}{
    short = CU,
    long = central unit,
    list = Central Unit,
    tag = abbrev
}

\DeclareAcronym{D2D}{
    short = D2D,
    long = device-to-device,
    list = Device-to-Device,
    tag = abbrev
}

\DeclareAcronym{DE-ADMM}{
    short = DE-ADMM,
    long = direct estimation with alternating direction method of multipliers,
    list = Direct Estimation with Alternating Direction Method of Multipliers,
    tag = abbrev
}

\DeclareAcronym{DE-BR}{
    short = DE-BR,
    long = direct estimation with best response,
    list = Direct Estimation with Best Response,
    tag = abbrev
}

\DeclareAcronym{DE-SG}{
    short = DE-SG,
    long = direct estimation with stochastic gradient,
    list = Direct Estimation with Stochastic Gradient,
    tag = abbrev
}

\DeclareAcronym{DFT}{
	short = DFT,
	long = discrete fourier transform,
	list = Discrete Fourier Transform,
	tag = abbrev
}

\DeclareAcronym{DoF}{
    short = DoF,
    long = degrees of freedom,
    list = Degrees of Freedom,
    tag = abbrev
}

\DeclareAcronym{DL}{
    short = DL,
    long = downlink,
    list = Downlink,
    tag = abbrev
}

\DeclareAcronym{GD}{
	short = GD, 
	long = gradient descent,
	list = Gradeitn Descent,
	tag = abbrev
}

\DeclareAcronym{IBC}{
    short = IBC,
    long = interfering broadcast channel,
    list = Interfering Broadcast Channel,
    tag = abbrev
}

\DeclareAcronym{i.i.d.}{
    short = i.i.d.,
    long = independent and identically distributed,
    list = Independent and Identically Distributed,
    tag = abbrev
}

\DeclareAcronym{JP}{
    short = JP,
    long = joint processing,
    list = Joint Processing,
    tag = abbrev
}

\DeclareAcronym{KKT}{
    short = KKT,
    long = Karush-Kuhn-Tucker,
    tag = abbrev
}

\DeclareAcronym{LOS}{
	short = LOS,
	long = line-of-sight,
	list = Line-of-Sight,
	tag = abbrev
}

\DeclareAcronym{LS}{
    short = LS,
    long = least squares,
    list = Least Squares,
    tag = abbrev
}

\DeclareAcronym{LTE}{
    short = LTE,
    long = Long Term Evolution,
    tag = abbrev
}

\DeclareAcronym{LTE-A}{
    short = LTE-A,
    long = Long Term Evolution Advanced,
    tag = abbrev
}

\DeclareAcronym{MIMO}{
    short = MIMO,
    long = multiple-input multiple-output,
    list = Multiple-Input Multiple-Output,
    tag = abbrev
}

\DeclareAcronym{MISO}{
    short = MISO,
    long = multiple-input single-output,
    list = Multiple-Input Single-Output,
    tag = abbrev
}

\DeclareAcronym{MSE}{
    short = MSE,
    long = mean-squared error,
    list = Mean-Squared Error,
    tag = abbrev
}

\DeclareAcronym{MMSE}{
    short = MMSE,
    long = minimum mean-squared error,
    list = Minimum Mean-Squared Error,
    tag = abbrev
}

\DeclareAcronym{mmWave}{
	short = mmWave,
	long = millimeter wave,
	list = Millimeter Wave,
	tag = abbrev
}

\DeclareAcronym{MU-MIMO}{
    short = MU-MIMO,
    long = multi-user \ac{MIMO},
    list = Multi-User \ac{MIMO},
    tag = abbrev
}

\DeclareAcronym{OTA}{
    short = OTA,
    long = over-the-air,
    list = Over-the-Air,
    tag = abbrev
}

\DeclareAcronym{PSD}{
    short = PSD,
    long = positive semidefinite,
    list = Positive Semidefinite,
    tag = abbrev
}

\DeclareAcronym{QoS}{
	short = QoS,
	long = quality of service,
	list = Quality of Service,
	tag = abbrev
}

\DeclareAcronym{RCP}{
	short = RCP,
	long = remote central processor,
	list = Remote Central Processor,
	tag = abbrev
}

\DeclareAcronym{RRH}{
    short = RRH,
    long = remote radio head,
    list = Remote Radio Head,
    tag = abbrev
}

\DeclareAcronym{RSSI}{
    short = RSSI,
    long = received signal strength indicator,
    list = Received Signal Strength Indicator,
    tag = abbrev
}

\DeclareAcronym{RX}{
	short = RX,
	long = receiver,
	list = Receiver,
	tag = abbrev
}

\DeclareAcronym{SCA}{
    short = SCA,
    long = successive convex approximation,
    list = Successive Convex Approximation,
    tag = abbrev
}

\DeclareAcronym{SG}{
    short = SG,
    long = stochastic gradient,
    list = Stochastic Gradient,
    tag = abbrev
}

\DeclareAcronym{SNR}{
    short = SNR,
    long = signal-to-noise ratio,
    list = Signal-to-Noise Ratio,
    tag = abbrev
}

\DeclareAcronym{SINR}{
    short = SINR,
    long = signal-to-interference-plus-noise ratio,
    list = Signal-to-Interference-plus-Noise Ratio,
    tag = abbrev
}

\DeclareAcronym{SOCP}{
	short = SOCP, 
	long = second order cone program,
	list = Second Order Cone Program,
	tag = abbrev
}

\DeclareAcronym{SSE}{
    short = SSE,
    long = stream specific estimation,
    list = Stream Specific Estimation,
    tag = abbrev
}

\DeclareAcronym{SVD}{
	short = SVD,
	long = singular value decomposition,
	list = Singular Value Decomposition,
	tag = abbrev
}

\DeclareAcronym{TDD}{
	short = TDD,
	long = time division duplex,
	list = Time Division Duplex,
	tag = abbrev
}

\DeclareAcronym{TX}{
	short = TX,
	long = transmitter,
	list = Transmitter,
	tag = abbrev
}

\DeclareAcronym{UE}{
    short = UE,
    long = user equipment,
    list = User Equipment,
    tag = abbrev
}

\DeclareAcronym{UL}{
    short = UL,
    long = uplink,
    list = Uplink,
    tag = abbrev
}

\DeclareAcronym{ULA}{
	short = ULA,
	long = uniform linear array,
	list = Uniform Linear Array,
	tag = abbrev
}

\DeclareAcronym{UPA}{
    short = UPA,
    long = uniform planar array,
    list = Uniform Planar Array,
    tag = abbrev
}

\DeclareAcronym{WMMSE}{
    short = WMMSE,
    long = weighted minimum mean-squared error,
    list = Weighted Minimum Mean-Squared Error,
    tag = abbrev
}

\DeclareAcronym{WMSEMin}{
    short = WMSEMin,
    long = weighted sum \ac{MSE} minimization,
    list = Weighted sum \ac{MSE} Minimization,
    tag = abbrev
}

\DeclareAcronym{WBAN}{
	short = WBAN,
	long = wireless body area network,
	list = Wireless Body Area Network,
	tag = abbrev
}

\DeclareAcronym{WSRMax}{
    short = WSRMax,
    long = weighted sum rate maximization,
    list = Weighted Sum Rate Maximization,
    tag = abbrev
}

\pgfplotsset{compat=1.17}

\usetikzlibrary{matrix,decorations.pathreplacing,arrows.meta,calc,backgrounds}

\begin{document}

\title{
Asymmetric Stream Allocation and Linear Decodability in MIMO Coded Caching
}



\author{\IEEEauthorblockN{Mohammad NaseriTehrani,
~\IEEEmembership{Student~Member,~IEEE,}  
MohammadJavad Salehi,~\IEEEmembership{Member,~IEEE,} 
~and Antti T\"olli}~\IEEEmembership{Senior~Member,~IEEE}
\thanks{
The authors are affiliated with the University of Oulu, Finland. Emails: 
\{mohammad.naseritehrani, mohammadjavad.salehi, antti.tolli\}@oulu.fi
}
}

\maketitle

\renewcommand{\thesubfigure}{Fig.~\arabic{subfigure}}


\begin{abstract}


Coded caching (CC) can transform cache memory at network devices into an active communication resource and significantly enhance the Degrees of Freedom (DoF) of multi-input multi-output (MIMO) systems by jointly exploiting global caching and spatial multiplexing gains. Existing linearly decodable MIMO-CC designs, however, largely rely on symmetric stream allocation, where all scheduled users receive the same number of streams, which induces coarse DoF granularity and may leave spatial dimensions unused. This letter studies one-shot linearly decodable MIMO-CC delivery with arbitrary per-user stream allocations. We derive a sufficient stream-count decodability condition, expressed through per-user stream counts and multicast-codeword multiplicities, that generalizes the symmetric common-stream feasibility rule. Building on this condition, we develop a greedy multicast scheduling procedure with certified linear decodability, which redistributes coded multicast messages across transmission intervals to realize asymmetric stream allocations. Numerical results show that the proposed scheduler fills DoF-granularity gaps and improves finite-SNR symmetric rates over the state of the art.

\end{abstract}


\section{Introduction}
\label{section:intro}
Coded caching (CC)~\cite{maddah2014fundamental} turns users' cache memories into an active communication resource by multicasting carefully designed codewords to groups of $t\!+\!1$ users, yielding a global caching gain $t$. In multiple-input single-output (MISO) systems, multicast beamforming serves multiple user groups while suppressing inter-group interference, combining caching and spatial multiplexing gains $L$ to achieve $t\!+\!L$ degrees of freedom (DoF)~\cite{shariatpanahi2016multi,tolli2017multi}.\color{black}

While MISO-CC is well-explored, the application of CC in multiple-input multiple-output (MIMO) systems has gained attention only recently. The optimal DoF and bounds for MIMO-CC schemes were studied in~\cite{cao2017fundamental} and~\cite{cao2019treating}, using complex interference alignment techniques. 
In~\cite{salehi2021MIMO}, insights from the shared-cache model~\cite{parrinello2019fundamental} were used to design a linear MIMO-CC solution achieving a DoF of $Gt\!+\!L$ with minimal subpacketization overhead, where $G$ is the number of receive antennas per user. The scheme, however, worked only under the restrictive conditions $\frac{L}{G}\!\in\!\mathbb{N}$ and $\frac{L}{G}\!\geq \!t$.
%
Later,~\cite{tehrani2024enhanced,naseritehrani2024multicast} generalized single-shot MIMO-CC and improved the DoF over~\cite{salehi2021MIMO} by jointly adjusting parallel streams decoded by each user and the number of target users per transmission, denoted by $\beta$ and $\Omega$, respectively. 
 Using signal-domain interference cancellation,~\cite{zhao2026vector}
exploits multi-antenna receivers through vector coded caching and
optimized power allocation.
\color{black}
More recently,~\cite{naseritehrani2026cache} proposed a bit-domain (XOR-based) delivery framework by introducing a new multicast scheduling scheme that improved practicality by ensuring linear decodability and bit-domain interference cancellation.
This scheduling enhanced finite-signal-to-noise ratio (SNR) performance by allowing the parameter $\beta$ to be flexibly selected from a set of feasible values, corresponding to $\Omega$ and denoted by $\CB_\Omega$.
%
%
However, the structure of the predefined set $\CB_\Omega$ restricted the number of admissible stream allocations, since symmetry across the allocated streams must be preserved, which can leave spatial dimensions unused.

This paper proposes a generalized one-shot MIMO-CC delivery framework that characterizes user-specific stream allocation under the linear decodability constraint, thereby enabling a broader class of bit-level multicast transmission designs beyond symmetric stream allocation in~\cite{naseritehrani2026cache} and filling the DoF-granularity gap therein. 
First, we derive a per-codeword linear decodability criterion based on per-user stream counts and multicast-codeword multiplicities, allowing both symmetric and asymmetric stream allocations. 
We then develop an asymmetric multicast scheduling scheme that redistributes coded multicast messages across transmission intervals to realize asymmetric per-user stream allocation while preserving linear decodability. 
Numerical results compare the proposed scheduler with symmetric MIMO-CC and conventional multiuser (MU)-MIMO baselines, quantifying the DoF-region enlargement, and finite-SNR symmetric-rate gains.
\color{blue}
\color{black}
\textbf{Notation}: $(.)\herm$ denotes Hermitian. For $J\!\!\in\!\!\mathbb N$ and $x\!\!\in\!\!\mathbb R$, $[J] \equiv \{1,\ldots,J\}$ and $[x]^{+}\!=\!\max\{x,0\}$. Boldface upper- and lower-case letters denote matrices and vectors, and calligraphic letters denote sets. For a set $\CK$, $|\CK|$ is its cardinality and $\CK\backslash\CT$ its set difference from $\CT$. 
$\SfB$ is a super-set and $|\SfB|$ its size; super-sets are multi-sets, allowing repeated entries. 

\section{System Model}
\label{section:sys_model}
We consider a MIMO system where a base station (BS) with $L$ transmit antennas serves $K$ cache-enabled users, each with $G$ receive antennas. The library $\CF$ contains $N$ equal-sized files, and each user can cache $M$ files in its memory, yielding the global CC gain $t\triangleq KM/N$, i.e., the aggregate normalized cache size. During placement, each file $W\in\CF$ is split into
$\binom{K}{t}$ subfiles $W_{\CP}$, where
$\CP\subseteq[K]$, $|\CP|=t$, and $W_{\CP}$ is cached by every user $k\in\CP$~\cite{maddah2014fundamental}.

At the beginning of the delivery phase, each user $k$ reveals its requested file $W_k \in \CF$ to the BS. For each subset of users $\CK(j)$ with $|\CK(j)| = \Omega$, $j\in[J]$, the BS constructs and transmits (e.g., in consecutive time slots)  a set of 
$I$ transmission vectors $\Bx(j,i)$, each delivering parts of the requested data to every user $k \in \CK(j)$. This results in a total number of $J \times I$ time intervals, where $J = \binom{K}{\Omega}$, and $I$ and $\Omega$ are defined by the delivery algorithm. 

Delivery may further split each packet $W_{\CP}$ into smaller subpackets
$W_{\CP}^{q}$, where $q$ distinguishes them
to avoid duplicate
transmissions and is omitted for simplicity. 
\color{black}
In time slot $i \in [I]$ for user set $\CK(j)$, the transmission vector $\Bx(j,i)$ is constructed as 
\begin{equation}
\label{eq:general_trans_vector}
\begin{aligned}
    \Bx(j,i) = \sum\nolimits_{\CT \in \SfS(j,i)} \Bw_{\CT} X_{\CT},
    \end{aligned}
\end{equation}
%
where $X_{\CT} := \bigoplus_{k \in \CT} W_{\CT \backslash \{k\},k}$ represents a codeword generated for a user subset $\CT \subseteq \CK(j)$ with $|\CT| = t+1$ ($\oplus$ denotes the bit-wise XOR operation), $\SfS(j,i)$ is the super-set of codeword indices in time slot $i$,
and $\Bw_{\CT}$ is the multicast beamforming (precoder) vector 
for $X_{\CT}$.
After the transmission of $\Bx(j,i)$, user $k \in \CK(j)$ receives $\By_k(j,i) = \BH_k(j,i) \Bx(j,i) + \Bz_k(j,i)$,
where $\BH_k(j,i)\in\mathbb{C}^{G \times L}$ is the block-fading channel matrix between the BS and user $k$ in interval $i$ with i.i.d. entries drawn from $\mathcal{CN}(0, 1)$, and $\Bz_k(j,i)\sim\mathcal{CN}(\mathbf{0}, N_0 \mathbf{I})$ represents the noise. 
The full channel state information (CSI) is assumed to be available at the BS.
%

Let $\Theta$ denote the final \emph{subpacketization}, including splitting factors from both placement and delivery phases. 
Using $R(j,i)$ (file/second) to denote the max-min transmission rate of $\Bx(j,i)$ that enables successful decoding at every user $k \in \CK(j)$, the transmission time of $\Bx(j,i)$ is $T(j,i) = \nicefrac{1}{(\Theta R(j,i))}$ seconds. The total delivery time is then $T_{\mathrm{total}} = \sum_{j \in [J]}\sum_{i\in[I]} T(j,i)$, and the symmetric rate is defined as $R_{\mathrm{sym}} = \nicefrac{K}{T_{\mathrm{total}}}$ (files/second).
Throughout the rest of the paper, we consider transmission intervals for a particular subset $\CK(j)$ and drop the $j$ index. The same operation is repeated for all other subsets of users with size $\Omega$.



\textbf{Reference symmetric scheduling.}
Let $\SfS^{\CK}$ denote the set of all the $t+1$-subsets of the target user subset $\CK$, so each user $k \in \CK$ appears in $\beta_{\mathrm{MC}} = \binom{\Omega-1}{t}$ sets in $\SfS^{\CK}$. One can construct a single-interval transmission ($I=1$) using all
codeword indices in $\SfS^{\CK}$ as
$\sum_{\CT\in\SfS^{\CK}}X_{\CT}\Bw_{\CT}$ and multicast it to every
$k\in\CK$. However, if $G<\beta_{\rm MC}$, decoding requires
computationally demanding successive interference cancellation (SIC),
whereas if $G\gg\beta_{\rm MC}$, far fewer streams are decoded than
the available Rx dimensions.
To address these issues, the hypergraph-based scheduling in~\cite{naseritehrani2026cache} flexibly selects the number of parallel decoded streams per user in each interval, \(\beta\), from the feasible set
%
\begin{equation}
\begin{aligned}
   \CB_{\Omega} = \Big\{ \hat{\eta} \hat{\beta} \mid \frac{\hat{\delta} \hat{S}}{\hat{\eta}} \in \mathbb{N}, \hat{\eta}  \leq \min (\eta_{\rm Tx},  \eta_{\rm Rx}) \Big\},
\end{aligned} 
\label{eq:lin_fullmc_mac}
\end{equation}
where
\begin{equation}
\label{eq:eta_def}
    \eta_{\rm Tx} = \frac{L \hat{S}}{1 + (\Omega - t - 1) \hat{S} \hat{\beta}},\qquad
    \eta_{\rm Rx} = \frac{G}{\hat{\beta}}
\end{equation}
represent 
the transmit-side null-space and receive-dimension limits, respectively, to ensure linear decodability, and are imposed by the symmetric stream allocation structure.
%
%
In~\eqref{eq:lin_fullmc_mac}, $\hat{\eta}$ and $\hat{\delta}$ are two general integers, and
\begin{equation}
\begin{aligned}\label{eq:lin_fullmc_mac1}
    \hspace{-2.mm}\hat{\beta} = \frac{t+1}{\mathrm{gcd}(t+1,\Omega)}, \quad \hat{B} = \frac{\Omega}{\mathrm{gcd}(t+1,\Omega)}, \quad \hat{S} = \frac{\binom{\Omega}{t+1}}{\hat{B}}.
\end{aligned}
\end{equation}
The scheduling in~\cite{naseritehrani2026cache} first organizes the
codeword indices in $\SfS^\CK$ into a $\hat B\times\hat S$ table
$\mathtt T_1$, with each user appearing in exactly $\hat\beta$ indices
per column. Concatenating $\hat{\delta}$ copies of $\mathtt{T}_1$
forms the $\hat{B}\times(\hat{\delta}\hat{S})$ table $\mathtt{T}_2$.
The $B\times S$ reference table $\mathtt{T}_{\rm ref}$, with
$B=\hat{\eta}\hat{B}$ and
$S=\frac{\hat{\delta}\hat{S}}{\hat{\eta}}$, is then obtained by
merging the codeword indices of every $\hat{\eta}$ consecutive columns
of $\mathtt{T}_2$. Hence, each user $k\in\CK$ appears in exactly
$\beta=\hat{\eta}\hat{\beta}$ indices per column and decodes $\beta$
streams when each column is transmitted in a separate interval
($I=S$).
 \color{black}
While this scheduling provides a flexible linear MIMO-CC solution,
the constraint imposed by
$\min(\eta_{\rm Tx},\eta_{\rm Rx})$ in the definition of $\CB_{\Omega}$ in~\eqref{eq:lin_fullmc_mac} 
restricts the symmetric scheduler to a coarse set of common-\(\beta\) operating points, and can leave gaps between feasible stream-allocation values.
For example, when \((L,G,t,\Omega)\!=\!(11,8,2,4)\), the feasible set reduces to \(\beta\! \in\! \CB_{\Omega=4}\! =\! \{3,6\}\), thereby excluding \(\beta \! \in \!\{4,5,7,8\}\).
This structural limitation can restrict the feasible operating points under symmetric allocation, which can restrict both the achievable DoF and symmetric rate performance optimized over \(\beta \in \CB_\Omega\). 
%
\color{black}
To overcome this issue, we relax the common-stream constraint and construct asymmetric schedules by redistributing the multicast indices of the symmetric reference table in~\cite{naseritehrani2026cache}, while preserving linear decodability.\color{black}
\section{Linear Decodability for MIMO-CC}
We next characterize a general linear decodability framework for the MIMO-CC transmission model in~\eqref{eq:general_trans_vector} as illustrated in Figure~\ref{fig:system_model}, directly in terms of per-user stream counts and multicast-codeword multiplicities, thereby supporting both symmetric and asymmetric per-user stream allocations.
Consider a general scheduling and a time interval $i \in [I]$. 
Let $\theta_{\CT}(i)$ denote the number of times a multicast group $\CT$ appears in $\SfS(i)$, and let  \(\beta_k(i)\) denote the number of parallel streams sent to a user $k \in \CK$ in this interval.

\begin{figure*}[t]
\centering
\begin{minipage}[t]{0.46\textwidth}
  \centering
  \includegraphics[height=3.4cm]{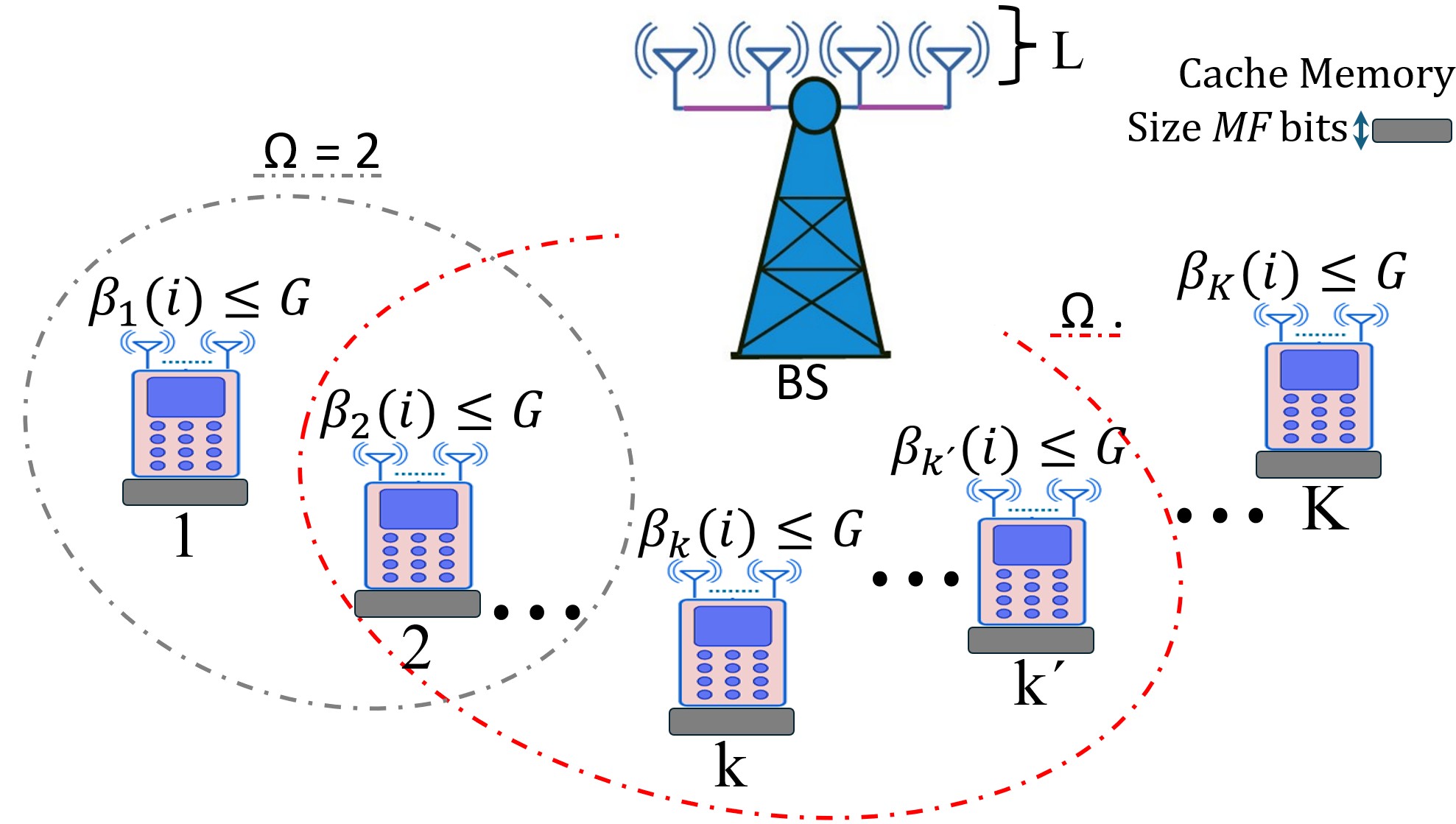}
  \caption{MIMO-CC scheduling: serving arbitrary $\Omega$ users per interval, each receiving asymmetric streams $\beta_k(i)\leq G$ per sub-interval.}
  \label{fig:system_model}
\end{minipage}\hfill
\begin{minipage}[t]{0.48\textwidth}
  \centering
  \resizebox{.8\linewidth}{!}{%
  \begin{tikzpicture}[x=1cm,y=1cm]

\newcommand{\BaseRedLineYOffset}{0.12cm}
\newcommand{\BaseRedLineXPad}{0.12cm}
\newcommand{\CHeaderYOffset}{0.32cm}     
\newcommand{\CTopBraceYOffset}{0.52cm}   
\newcommand{\CILabelXShift}{-0.82cm}     
\newcommand{\TableCellWidth}{0.82cm}     
\newcommand{\TableCellHeight}{0.53cm}
\newcommand{\AddCellHeight}{0.58cm}
\newcommand{\TableEntryFont}{\Large}

\tikzset{
  blk/.style={
    matrix of math nodes,
    nodes={
      minimum width=\TableCellWidth,
      minimum height=\TableCellHeight,
      inner sep=0pt,
      anchor=center,
      font=\TableEntryFont
    },
    column sep=0pt,
    row sep=0pt,
    ampersand replacement=\&
  },
  addblk/.style={
    matrix of math nodes,
    nodes={
      minimum width=\TableCellWidth,
      minimum height=\AddCellHeight,
      inner sep=0pt,
      anchor=center,
      font=\TableEntryFont
    },
    column sep=0pt,
    row sep=0pt,
    ampersand replacement=\&
  },
  tableline/.style={draw=black!65, line width=0.25pt},
  lab/.style={font=\large},
  biglab/.style={font=\LARGE},
  braceblue/.style={cyan!80!blue, line width=0.7pt, decorate, decoration={brace, amplitude=3pt}},
}

\newcommand{\X}[1]{X_{\scriptstyle #1}}

\newcommand{\BaseTable}{
  \X{12} \& \X{14}\\
  \X{34} \& \X{23}\\
  \X{25} \& \X{15}\\
  \X{13} \& \X{24}\\
  \X{45} \& \X{35}\\
}

\newcommand{\TwoColFrame}[1]{%
  \draw[tableline] (#1.north west) rectangle (#1.south east);
  \draw[tableline]
    ($(#1.north west)!0.5!(#1.north east)$) --
    ($(#1.south west)!0.5!(#1.south east)$);
}

\newcommand{\ShadeBlock}[1]{%
  \begin{pgfonlayer}{background}
    \fill[blue!25] (#1.north west) rectangle (#1.south east);
  \end{pgfonlayer}
}

\matrix[blk] (A) at (0,0.20) {\BaseTable};
\TwoColFrame{A}

\draw[braceblue, decoration={brace, amplitude=3pt, mirror}]
  ($(A.north west)+(-0.10,0.05)$) --
  ($(A.south west)+(-0.10,-0.05)$)
  node[midway,left=0.20cm,black,lab] {$B=5$};

\draw[braceblue, decoration={brace, amplitude=3pt, mirror}]
  ($(A.south west)+(-0.04,-0.12)$) --
  ($(A.south east)+(0.04,-0.12)$)
  node[midway,below=0.08cm,black,lab] {$S=2$};

\node[lab] at (-1.55,-1.00) {$(a)$};
\node[biglab] at (0,-2.30) {$\mathbf{T}_{\mathrm{ref}}$};

\draw[fill=blue!10, draw=blue!45!black, line width=0.55pt]
  (1.52,-0.60) -- (2.00,-0.60) -- (2.00,-1.03) -- (2.82,0.18) --
  (2.00,1.39) -- (2.00,0.96) -- (1.52,0.96) -- cycle;
\foreach \xx in {1.30,1.38,1.46}
  \draw[blue!45!black,line width=0.45pt] (\xx,-0.34)--(\xx,0.70);

\node[biglab] at (3.25,1.05) {$\mathbf{T}_{\mathrm{base}}$};

\matrix[blk] (B1) at (5.15,0.35) {\BaseTable};
\TwoColFrame{B1}

\matrix[blk] (B2) at (6.95,0.35) {\BaseTable};
\TwoColFrame{B2}

\foreach \xx in {8.55,8.95,9.35,9.75,10.15,10.55}
  \node[biglab] at (\xx,0.35) {$\cdot$};

\matrix[blk] (B3) at (12.20,0.35) {\BaseTable};
\TwoColFrame{B3}

\draw[braceblue, decoration={brace, amplitude=3pt}]
  ($(B1.north west)+(-0.06,0.18)$) --
  ($(B3.north east)+(0.06,0.18)$)
  node[midway,above=0.03cm,black,lab] {$\delta \times S = 7 \times 2$};

\draw[white!90!red, line width=1.3pt, dash pattern=on 8pt off 4pt on 1.3pt off 4pt]
  ([xshift=-\BaseRedLineXPad,yshift=-\BaseRedLineYOffset]B1.south west) --
  ([xshift=\BaseRedLineXPad,yshift=-\BaseRedLineYOffset]B3.south east);

\node[lab] at (3.55,-1.00) {$(b)$};

\node[biglab] at (3.45,-4.30) {$\mathbf{T}_{\mathrm{new}}$};

\matrix[blk] (C12) at (5.80,-3.35) {\BaseTable};
\TwoColFrame{C12}

\node[lab] at ($(C12.north west)+(\CILabelXShift,\CHeaderYOffset)$) {$i=$};
\node[lab] at ($(C12.north west)!0.25!(C12.north east)+(0,\CHeaderYOffset)$) {$1$};
\node[lab] at ($(C12.north west)!0.75!(C12.north east)+(0,\CHeaderYOffset)$) {$2$};

\matrix[addblk] (D12) at (5.80,-5.65) {
  \X{14} \& \X{12}\\
  \X{23} \& \X{34}\\
};
\ShadeBlock{D12}
\TwoColFrame{D12}

\matrix[blk] (C34) at (8.25,-3.35) {\BaseTable};
\TwoColFrame{C34}

\node[lab] at ($(C34.north west)!0.25!(C34.north east)+(0,\CHeaderYOffset)$) {$3$};
\node[lab] at ($(C34.north west)!0.75!(C34.north east)+(0,\CHeaderYOffset)$) {$4$};

\matrix[addblk] (D34) at (8.25,-5.65) {
  \X{15} \& \X{13}\\
  \X{24} \& \X{25}\\
};
\ShadeBlock{D34}
\TwoColFrame{D34}

\node[biglab] at (10.25,-4.05) {$\cdots$};

\matrix[blk] (C910) at (12.25,-3.35) {\BaseTable};
\TwoColFrame{C910}

\node[lab] at ($(C910.north west)!0.25!(C910.north east)+(0,\CHeaderYOffset)$) {$9$};
\node[lab] at ($(C910.north west)!0.75!(C910.north east)+(0,\CHeaderYOffset)$) {$10$};

\matrix[addblk] (D910) at (12.25,-5.65) {
  \X{24} \& \X{13}\\
  \X{35} \& \X{45}\\
};
\ShadeBlock{D910}
\TwoColFrame{D910}

\draw[braceblue, decoration={brace, amplitude=3pt}]
  ($(C12.north west)+(-0.08,\CTopBraceYOffset)$) --
  ($(C910.north east)+(0.08,\CTopBraceYOffset)$)
  node[midway,above=0.02cm,black,lab] {$\tilde S = 10$};

\draw[braceblue, decoration={brace, amplitude=3pt, mirror}]
  ($(D12.north west)+(-0.10,0.04)$) --
  ($(D12.south west)+(-0.10,-0.04)$)
  node[pos=0.5,left=0.12cm,yshift=0.18cm,black,lab] {$m=2$};

\draw[braceblue, decoration={brace, amplitude=3pt}]
  ($(C910.north east)+(0.17,0.04)$) --
  ($(D910.south east)+(0.17,-0.04)$)
  node[midway,right=0.12cm,black,lab] {$\tilde B = 7$};

\node[lab] at (3.95,-6.05) {$(c)$};

\end{tikzpicture} 
  }
  \caption{Example~1: $\Omega$=5, $t$=1, $L$=10, $G$=3; (a) symmetric scheduling, (b) replicated symmetric table with $\delta$, (c) asymmetric scheduling table.}
  \label{fig2}
\end{minipage}
\vspace{-.5cm}
\end{figure*}
\begin{lemma}[Sufficient stream-count decodability condition]\label{Th:DoF_sch}
With the transmission model in~\eqref{eq:general_trans_vector}, linear decoding of $\beta_k(i)$ streams is possible at each user $k \in \CK$ in interval $i$
if the parameters \(\theta_{\mathcal{T}}(i)\) and \(\beta_k(i)\) 
satisfy
\begin{equation} 
\begin{aligned}\label{DoF_bndd}
   &{\text{C1}:} \sum\nolimits_{k'\in \CK\backslash \CT}\beta_{k'}(i) + \theta_{\CT}(i)\le L,\: \; \; \forall \CT \in \SfS(i), \\ 
   &{\text{C2}:} \beta_k(i) \leq G,\:  \:\quad\quad\quad\quad\quad \qquad \qquad \forall k\in \CK.
\end{aligned}
\end{equation}

\end{lemma}
\begin{proof}
By definition, the beamformer $\Bw_{\CT}$
must null out the inter-stream interference caused by $X_{\CT}$ to every stream decoded by each user $k' \in \CK\backslash \CT$ (as interference at users in $\CT$ is removed using cached content).
Let $\BU_k(i) \in \mathbb{C}^{G \times \beta_k(i)}$ denote the receive beamforming matrix used by user~$k$ to decode its $\beta_k(i)$ parallel streams in interval $i$. 
Define the equivalent interference channel of $X_{\CT}$ as 
\begin{equation}
\label{eq:equi_intf_chan}
\bar{\BH}_{\CT}(i)   = [\BH_{k'}(i)\herm\BU_{k'}(i)]\herm, \quad \forall k' \in \CK\backslash \CT,
\end{equation}
where $[\;\cdot\;]$ denotes horizontal concatenation. To eliminate inter-stream interference, $\Bw_{\CT}$ must lie in the nullspace of $\bar{\BH}_{\CT}(i)$, i.e.,
    $\Bw_{\CT} \in \mathrm{Null}(\bar{\BH}_{\CT}(i))$. 
%
%
Now, to ensure that each user~$k \in \CK$ can linearly decode all its $\beta_k (i) \leq G$ parallel streams, it is essential that the respective transmit beamformers are linearly independent. 
This condition can be satisfied for parallel streams $X_{\CT}$
with different multicast indices $\CT$, since their beamformers can
be selected linearly independently from the respective null spaces. \color{black}
%
%
However, for substreams sharing the same multicast stream index $\CT$, linear independence is limited by the dimensions of $\mathrm{Null}(\bar{\BH}_{\CT}(i))$. From~\eqref{eq:equi_intf_chan}, $\bar{\BH}_{\CT}(i)$ consists of $\Omega-t-1$ concatenated matrices $\BH_{k'}(i)\herm\BU_{k'}(i)$, each of dimension $L \times \beta_{k'}(i)$, yielding $\bar{\BH}_{\CT}(i) \in \mathbb{C}^{\sum_{k'\in \CK\backslash \CT}\beta_{k'}(i) \times L}.$
Applying the rank-nullity theorem
and using $\rm rank(\bar{\mathbf H}_{\mathcal T}(i))
\leq\sum_{k'\in\CK\setminus\CT}\beta_{k'}(i)$, \color{black}we obtain
\begin{equation*}
 \label{eq:rank-null-theorem} 
\mathrm{nullity}(\bar{\BH}_{\CT}(i) )  = L - \mathrm{rank}(\bar{\BH}_{\CT}(i) ) \geq L - \sum\nolimits_{k'\in \CK\backslash \CT}\beta_{k'}(i).
 \end{equation*}
%

Hence, $C1$ guarantees $\theta_{\CT}(i)\leq \rm nullity(\bar{\BH}_{\CT}(i))$, ensuring sufficient transmit null-space dimensions for the corresponding substreams sharing the same multicast index. Since the above rank bound can be strict, e.g., when receive combining reduces the effective interference-channel rank below the stream-count bound, $C1$ can be conservative\color{black}. Condition $C2$ follows from the Rx-antenna constraint $\beta_k(i)\leq G$.
%
%
Note that if $\beta_k(i) = \beta$ for all $k \in \CK$, condition $C_1$ reduces to $(\Omega-t-1)\beta+\nicefrac{\beta}{\binom{\Omega-1}{t}} \leq L$, which recovers the symmetric stream-allocation feasibility condition in~\cite{naseritehrani2026cache}. 
\end{proof}

\section{Asymmetric Multicast MIMO-CC Scheduling}\label{sec:scheduling}

For each target set $\CK$, we start from the $B\times S$ reference
scheduling table $\mathtt T_{\rm ref}$ in~\cite{naseritehrani2026cache}, with each column defining a separate transmission interval.
To enable asymmetric per-user stream allocation, we concatenate ${\delta}$ copies of the reference table $\mathtt{T}_{\mathrm{ref}}$ to obtain a $B \times ({\delta} S)$ enlarged table $\mathtt{T}_{\mathrm{base}}$. 
Then, retain $\tilde{S}$ columns, and decompose the remaining ${\delta} S - \tilde{S}$ columns. The multicast indices of the decomposed columns are redistributed beneath the retained $\tilde{S}$ columns, such that each column gets $m$ additional indices. This results in the final scheduling table $\mathtt{T}_{\mathrm{new}}$ with size $\tilde{B} \times \tilde S$, where $\tilde{B} = B+m$. The codeword indices in each column of $\mathtt{T}_{\mathrm{new}}$ are sent in a separate interval (i.e., $I = \tilde S$). Since the total number of multicast indices in the decomposed columns is $B \times ({\delta} S - \tilde{S})$, the parameters must satisfy
\begin{equation}
\label{eq:m_delta_s_rel}
    m \times \tilde{S} = B \times ({\delta} S - \tilde{S}),
 \end{equation}
 which yields $\tilde{S} \!\!=\! \frac{BS {\delta}}{B+m}$. Any integer triple $(m, {\delta}, \tilde{S})$ satisfying~\eqref{eq:m_delta_s_rel} defines a
$\tilde{B}\!\times\!\tilde{S}$ table $\mathtt{T}_{\rm new}$. Since the DoF of base scheduling, i.e., the total decoded parallel streams, is $\mathrm{DoF}_{\mathrm{ref}} =\! B(t+1) \!=\! \Omega \times \beta$, and each of the $m$ added multicast indices contributes $t+\!1$ streams, the resulting scheduling DoF is
%
\begin{equation}
\label{eq:dof_new_sch}
\mathrm{DoF_{sch}} = \sum\nolimits_{k \in \CK} \beta_k = \mathrm{DoF_{\mathrm{ref}}} + m \times (t+1).
\end{equation}
\begin{exmp}
\label{exmp:main_example}
Consider a system with a large user pool and $(L,G,t)=(10,3,1)$. Under the symmetric scheduling, the optimized \(\Omega = 5\) and \(\beta = 2 \in \CB_{\Omega=5}\) yield $\mathrm{DoF}_{\mathrm{ref}} = 10$. For the target set $\CK=\{1,\ldots,5\}$, Fig.~\ref{fig2}(a) shows the reference scheduling with $S=2$ sub-intervals and $B=5$ codeword indices per
column. Now, repeating this base scheduling (while increasing the subpacketization) by a factor of ${\delta} = 7$ yields the replicated table in Fig.~\ref{fig2}(b), consisting of $2 \times 7 = 14$ transmissions. Decomposing the last $4$ transmissions and redistributing their codeword indices among the remaining ones, as shown in Fig.~\ref{fig2}(c), results in a new scheduling with $10$ transmission vectors. In each transmission, $4$ users decode $\beta_k = 3$ and $1$ user decodes $\beta_k = 2$ streams, giving $\mathrm{DoF}_{\mathrm{sch}} \!=\! 14 \!=\! 10 + 2 \times 2$. 
Using 
Lemma~\ref{Th:DoF_sch}, \color{black}
linear decodability is verified and preserved in every sub-interval.
%
%

Keeping $(\Omega,t,G,\beta)\!=(5,1,3,2)$ fixed, for
$m\!=\!1$ $(\mathrm{DoF}_{\rm sch}\!=\!12)$, a 
joint Tx/Rx
feasibility search performed without imposing C1-C2 finds a feasible
solution at $L\!=\!8$, although C1 is violated. For $m\!=\!2$
$(\mathrm{DoF}_{\rm sch}\!=\!14)$, rank-$9$ 
interference rules out
$L\!=\!9$, while at $L\!=\!10$, C1 certifies the construction and
the Tx/Rx search finds a feasible solution.
\end{exmp}

Let $\tilde\SfS(i)$, $i\!\in\![\tilde S]$ denote the superset of multicast indices in the $i$-th column of the scheduling table $\mathtt{T}_{\mathrm{new}}$, and let $\beta_k(i)$ be the number of sets $\CT\!\! \in\!\! \tilde{\SfS}(i)$ that include $k$.
%
%

Since each added multicast index contributes \(t+1\) streams, 
C2 of Lemma~\ref{Th:DoF_sch} necessarily requires $m(t+1)\!\leq\! \Omega(G-\beta)$, which gives 
\(m\!\leq\! \lfloor \!\frac{(G - \beta)\Omega}{t+1} \rfloor\).
\color{black}
%
%
%
%
During the table decomposition--reassignment process, adding a multicast index ${\CT}$ to column $i$ updates
$\beta_k(i)\!\leftarrow \beta_k(i)+1$ for $k\!\in\!\CT$ and
$\theta_{\CT}(i)\!\leftarrow \theta_{\CT}(i)+1$. 
Hence, the reassignment should favor multicast indices with limited overlap while preserving C1\!-C2.\color{black}
%

\subsection{Asymmetric multicast scheduling design:}
\label{subsec:asym_scheduling_design}
%
%
For the decomposition--reassignment stage, we propose a greedy procedure that selects low-overlap multicast indices for each $\tilde{\SfS}(i)$ while accepting only reassignments satisfying
\color{black}Lemma~\ref{Th:DoF_sch}\color{black}.
%
%
First, we define a one-to-one function $\psi(\cdot)$, mapping each column $s\in[S]$ of the reference scheduling table $\mathtt{T}_{\rm ref}$ to a column index $\psi(s) = \bar{s}$, 
such that \color{black}$\bar{s} \in [S]$ and $\bar{s}\neq s$ for $S>1$  \color{black}(e.g., we can use a simple cyclic mapping).
Let us also define $\SfS_{\rm ref}(s)$ as the superset of all multicast indices in the column~$s$ of $\mathtt{T}_{\rm ref}$. Then, for each baseline column $s\in[S]$, we construct a superset $\SfC_{s}$ satisfying the following properties:
\begin{enumerate}
    \item $|\SfC_s| = d$, where $d$ is a fixed value,
    \item Each element $\SfA \in \SfC_{s}$ is an unordered set of $m$ distinct codeword indices selected from $\SfS_{\rm ref}(\bar{s})$,
    \item $\SfC_{s}$ is $\sigma$-regular over $\SfS_{\rm ref}(\bar{s})$, i.e., every multicast index $\CT\in\SfS_{\rm ref}(\bar{s})$ appears exactly in $\sigma$ sets in $\SfC_{s}$.
\end{enumerate}
Note that, as $|\SfS_{\rm ref}(\bar{s})| = B$, these properties imply that $\sigma B = dm$. 
\color{black}
For a selected \(m\), let \(g_m=\gcd(B,m)\). The smallest positive integer solution of \(dm=B\sigma\) is \(d=\nicefrac{B}{g_m}\) and \(\sigma=\nicefrac{m}{g_m}\), yielding the smallest table repetition;
if this minimum pair admits no regular C1-C2-feasible assignment, \(d\) and \(\sigma\) may be jointly scaled by a common positive integer factor. 
A feasible regular assignment can be used directly; otherwise, $\SfC_s$ is constructed by the greedy procedure in Algorithm~\ref{alg:balanced_greedy_swap}.
\color{black}
%
%
\color{black}Given the collections \(\SfC_s\), \(s\in[S]\), we form \color{black} the final scheduling table $\mathtt{T}_{\rm new}$ by initially concatenating $\delta = d+\sigma$ copies of $\mathtt{T}_{\rm ref}$ to obtain $\mathtt{T}_{\rm base}$. Then, we keep $\tilde{S} = dS$ columns of $\mathtt{T}_{\rm base}$ and decompose the rest $\sigma S$ columns. The multicast indices in the decomposed columns are distributed under retained columns with the help of $\SfC_s$ sets: each retained column $i \in [\tilde S]$ is a replica of a column $s \in [S]$ of $\mathtt{T}_{\rm ref}$, so at the end of column $i$, we append codeword indices in one distinct $\SfA \in \SfC_s$. 
This yields a $(B+m)\times\tilde S$ scheduling table $\mathtt{T}_{\rm new}$, with the codeword indices of each column transmitted in a separate time interval. \color{black}
Accordingly, $\delta$ is a positive-integer multiple of $\nicefrac{(B+m)}{g_m}$; while $\mathrm{DoF}_{\rm sch}$ grows linearly with $m$
by~\eqref{eq:dof_new_sch}, the corresponding $\delta$-fold
subpacketization increase w.r.t.~\cite{naseritehrani2026cache} is governed by $\gcd(B,m)$.\color{black}
\begin{exmp}\label{exmp2}
Consider a large network with $(L,G,t,\Omega)=(11,6,2,5)$. Under symmetric scheduling, with 
$(\beta,S,B)=(3,2,5)$ and $\beta\in\CB_{\Omega=5}$, we get $\mathrm{DoF}_{\mathrm{ref}}=15$.
With the proposed scheduling, selecting $m=3$ increases the DoF to
$\mathrm{DoF}_{\mathrm{sch}}=24$.
Let the baseline columns $\SfS_{\rm ref}(1)\!=\!\{123,124,345,235,145\}$ and $\SfS_{\rm ref}(2)\!=\!\{125,134,234,245,135\}$, with multicast group brackets omitted for simplicity.
A valid choice of supersets is obtained as  
\begin{equation*}
       \small \begin{aligned}
             \SfC_{s=1} &= \big\{\{125,134,234\},\{125,234,135\},\\& \quad \{134,245,125\}, \{134,245,135\},\{234,135,245\}\big\}\\
             \SfC_{s=2}  &= \big\{\{123,345,124\},\{123,145,235\},\\& \quad \{124,345,235\}, \{124,235,145\},\{345,123,145\}\big\}.
             \end{aligned}
\end{equation*}
Here $g_m=1$, hence $d=5$, and each
$\CT\in\SfS_{\rm ref}(\bar{s})$ appears in exactly $\sigma=3$ elements
of $\SfC_s$. Accordingly, $\delta=8$ copies of $\mathtt{T}_{\rm ref}$ are concatenated, $\sigma S=6$ columns are decomposed, and each of the five retained replicas of $\SfS_{\rm ref}(1)$ is appended by multicast indices in one distinct element of $\SfC_{s=1}$.
\end{exmp}

\subsection{Constructing $\SfC_s$ supersets}
\label{subsec:constructing_Cs}
To construct $\SfC_s$ sets under low-overlap assignment and linear decodability constraints, we adopt the balanced greedy selection procedure in Algorithm~\ref{alg:balanced_greedy_swap}. For each target column $s$, the algorithm generates $d$ sets of multicast indices using the elements in a \emph{donor} column $\bar{s} = \psi(s)$. Each set, denoted by $\SfA$, is formed by sequentially selecting $m$ multicast indices $\CT$ from a dynamically updated feasible set $\SfF$. At each iteration, $\SfF$ contains those indices $\CT$ whose inclusion in $\SfA$ preserves linear decodability (verified by the \textsc{LinearFeasibleCheck} function), satisfies the pairwise overlap constraint with threshold $\tau$, and does not violate the prescribed regularity level $\sigma$, enforced via the quota variable $q(\CT)$.

If 
$\SfF=\emptyset$, a repair step swaps multicast indices between the current set \(\SfA\) and a previously constructed set \(\SfB\), while maintaining both the overlap and C1-C2 conditions. \color{black}The parameter \(J_{\max}\) limits the search effort for constructing each set. If a valid collection \(\SfC_s\) is not found, the search may be repeated with a larger \(J_{\max}\) or another
admissible \(\tau\in[[2(t+1)-\Omega]^+,t]\); such a failure only indicates that the current greedy search did not find a feasible assignment. For fixed \((m,d,\sigma,\tau,J_{\max})\), the procedure terminates after finitely many checks, and any returned schedule is certified linearly decodable by Lemma~\ref{Th:DoF_sch}.
\color{black}

\begin{algorithm}[t]
\caption{Balanced greedy selection for column $s \in [S]$}
\label{alg:balanced_greedy_swap}
\small
\begin{algorithmic}[1]
        \State $\SfC_{s} \gets \varnothing, \quad \bar{s}\gets\psi(s), \quad B \gets |\SfS_{\rm ref}(\bar{s})|, \quad \sigma \gets \frac{d \cdot m}{B}$
        \ForAll{$\CT \in \SfS_{\rm ref}(\bar{s})$}
             $q(\CT) \gets \sigma$
        \EndFor
        \While{$|\SfC_{s}| < d$}
            \State $\SfA \gets \varnothing, \quad j \gets 0$
            \While{$|\SfA| < m, j < J_{\max}$}
                \State $j \gets j+1$
                \State $\SfP \gets$ \textsc{LinearFeasibleCheck}($s,\bar{s},\SfA$)
                \State $\SfF \gets \{\CT \in \SfP : q(\CT) > 0, \: |\CT'\cap\CT|\leq\tau,\: \forall\,\CT'\in\SfA\}$
                \If{$\SfF = \varnothing$}
                    \State Choose $\SfB \in \SfC_{s}, \CT_B \in \SfB, \CT_A \in \SfA$
                    \State $\tilde{\SfA}\gets \SfA\setminus \color{black}\{\CT_A\}\color{black},\quad
                          \tilde{\SfB}\gets \SfB\setminus\color{black}\{\CT_B\}\color{black}$
                    \If{$\max\limits_{\CT \in \tilde{\SfA}} |\CT \cap \CT_B| \le \tau, \max\limits_{\CT \in \tilde{\SfB}} |\CT \cap \CT_A| \le \tau$}
                        \State $\SfA \gets \tilde{\SfA} \cup \{\CT_B\}$
                        \State $\SfB \gets \tilde{\SfB} \cup \{\CT_A\}$
                        \State \textbf{continue}
                    \EndIf                       
                \EndIf
                \State $\CT^\star \gets \arg\min_{\CT \in \SfF} \beta \sum_{\CT' \in \SfA} |\CT' \cap \CT| - q(\CT)$ 
                \State $\SfA \gets \SfA \cup \{\CT^\star\}$, \quad $q(\CT^\star) \gets q(\CT^\star) - 1$
            \EndWhile
            \If{$|\SfA| < m$}
                \textbf{break}
            \EndIf
            \State $\SfC_{s} \gets \SfC_{s} \cup \{\SfA\}$
        \EndWhile
\end{algorithmic}
\end{algorithm}
\begin{algorithm}[t]
\caption{Linear Decodability Check Function}
\small
\begin{algorithmic}[1]
  \small  \Function{LinearFeasibleCheck}{$s,\bar{s},\SfA$}
        \State $\SfP \gets \varnothing$
        \ForAll{$\CT \in \SfS_{\rm ref}(\bar{s})$}
            \ForAll{$k \in \CK$} $b(k) \gets 0$ \EndFor
            \State $\SfU \gets \SfS_{\rm ref}(s) \cup \SfA \cup \{\CT\}$
            \ForAll{$\CT' \in \SfU$} $c(\CT') \gets 0$ \EndFor
            \ForAll{$\CT' \in \SfU$} 
                \State $c(\CT') \gets c(\CT') + 1$
                \ForAll{$k \in \CT'$} $b(k) \gets b(k) + 1$ \EndFor
            \EndFor
            \ForAll{$\CT' \in \SfU$} $n(\CT') \gets c(\CT') + \sum_{k \in \CK \setminus \CT'} b(k)$ \EndFor
            \If{$\max\limits_{k \in \CK} b(k) \le G$ and $\max\limits_{\CT' \in \SfU} n(\CT') \le L$}
                \State $\SfP \gets \SfP \cup \{\CT\}$
            \EndIf
        \EndFor
        \State \Return $\SfP$
    \EndFunction
\end{algorithmic}\label{alg:LD_alg}
\end{algorithm}

\section{Simulation Results}
\label{section:Simulations}
\begin{figure*}[!t]
\centering
\begin{minipage}[t]{0.32\textwidth}
\centering
\input{Figs/Plot_scal_cc_gain_single}
{\footnotesize (a)}
\end{minipage}\hfill
\begin{minipage}[t]{0.32\textwidth}
\centering
\input{Figs/Plot_scal_RX_SM} 
{\footnotesize (b)}
\end{minipage}\hfill
\begin{minipage}[t]{0.32\textwidth}
\centering
    \centering
    \resizebox{0.95\columnwidth}{!}{%


    \begin{tikzpicture}

    \begin{axis}
    [
    axis lines = center,
    xlabel near ticks,
    xlabel = \smaller {SNR [dB]},
    ylabel = \smaller {Symmetric Rate [file/s]},
    ylabel near ticks,
    ymin = 17, 
    xmax = 30,
    legend style={
        nodes={scale=0.9, transform shape},
        at={(0,1)}, 
        anchor=north west,
        draw=black, 
        outer sep=2pt, 
        font=\tiny, 
    },
    ticklabel style={font=\smaller},
    grid=both,
    major grid style={line width=.2pt,draw=gray!30},
    ]



\addplot
    [dashed, line width=1.pt, mark = triangle, mark size=4pt,mark options={solid, fill=black!90}, black!90]
    table[y=RsymproposedL11G8t0omg11DoF11,x=SNR]{Figs/data_L11_G8.tex};
    \addlegendentry{\footnotesize MU-MIMO
    }

    \addplot
    [dashed, line width=1.pt, mark = star, mark size=4pt,mark options={solid, fill=red!90}, red!100]
    table[y=RsymproposedL11G8t2omg4DoF12,x=SNR]{Figs/data_L11_G8.tex};
    \addlegendentry{\footnotesize $\Omega$=4, {DoF}=12~\cite{naseritehrani2026cache}}

    \addplot
    [dashed,line width=1.pt,  mark = star, mark size=4pt , blue!90]
    table[y=RsymproposedL11G8t2omg4DoF15,x=SNR]{Figs/data_L11_G8.tex};
    \addlegendentry{\footnotesize $\Omega$=4, {DoF}=15}
    \addplot
    [mark = square,line width=1.pt, mark size=4pt, blue!90]
    table[y=RsymproposedL11G8t2omg4DoF18,x=SNR]{Figs/data_L11_G8.tex};
    \addlegendentry{\footnotesize $\Omega$=4, {DoF}=18}
    \addplot
    [dash dot,line width=1.pt,  mark = o , mark size=4pt,mark options={solid, fill=blue!90}, blue!90]
    table[y=RsymproposedL11G8t2omg4DoF21,x=SNR]{Figs/data_L11_G8.tex};
    \addlegendentry{\footnotesize $\Omega$=4, {DoF}=21}

\addplot
    [dashed, line width=1.pt, mark = o, mark size=4pt,mark options={solid, fill=red!90}, red!100]
    table[y=RsymproposedL11G8t2omg4DoF24,x=SNR]{Figs/data_L11_G8.tex};
    \addlegendentry{\footnotesize $\Omega$=4, {DoF}=24~\cite{naseritehrani2026cache}}
    \addplot
    [dashed, line width=1.pt, mark = +, mark size=4pt ,mark options={solid, fill=blue!90}, blue!90]
    table[y=RsymproposedL11G8t2omg4DoF27,x=SNR]{Figs/data_L11_G8.tex};
    \addlegendentry{\footnotesize $\Omega$=4, {DoF}=27}
    \addplot
    [dash dot,line width=1.pt,  mark = x, mark size=4pt ,mark options={solid, fill=blue!90}, blue!90]
    table[y=RsymproposedL11G8t2omg4DoF30,x=SNR]{Figs/data_L11_G8.tex};
    \addlegendentry{\footnotesize $\Omega$=4, {DoF}=30}

    \end{axis}

    \end{tikzpicture}
    }
    \label{fig:plot_7}
{\footnotesize (c)}
\end{minipage}
\caption{Performance comparison of the proposed asymmetric scheduling. (a) DoF granularity enhancement, $L=11$, $G=8$, $t\in\{1,2,3\}$, $\Omega\in\{4,5\}$. (b) DoF granularity and maximum DoF enhancement, $L=11$, $t=2$, $\Omega=4$, $G\in\{5,8,11\}$. (c) Rate performance for $(L,G,t)=(11,8,2)$.}
\label{fig:performance_comparison}
\vspace{-0.5cm}
\end{figure*}
We use numerical results to validate the effectiveness of the proposed scheduling.
Unless stated otherwise, $K=20$. $\mathrm{SNR}=P/N_0$, where $P$ is the available BS transmit power. Rates are evaluated using the max--min formulation in~\cite[Eq.~(36)]{naseritehrani2026cache}, under the same channel distribution and power constraint for all schemes.

%
%
For different $(t,\Omega)$ at fixed $(L,G)\!=\!(11,8)$,
Fig.~\ref{fig:performance_comparison}(a) shows that the proposed greedy
scheduling enlarges the achievable DoF set over the symmetric scheduling
in~\cite{naseritehrani2026cache}, \color{black}with no tight converse yet available. \color{black}
By relaxing the symmetry constraint, the proposed solution 
adapts across these configurations and fills DoF gaps unavailable to the symmetric scheduling, enabling new per-user stream allocations beyond the coarse common-stream operating points induced by \(\CB_\Omega\) in~\cite{naseritehrani2026cache}.
Fig.~\ref{fig:performance_comparison}(b) shows the DoF enlargement 
as the number of Rx antennas \(G\) scales. 
The symmetric scheduling in~\cite{naseritehrani2026cache} saturates once the common stream count is limited by the Tx-side null-space constraint rather than the Rx dimension, i.e., when \color{black}$\lfloor\eta_{\rm Rx}\rfloor\geq\lfloor\eta_{\rm Tx}\rfloor$ \color{black} in~\eqref{eq:eta_def}.
Hence, increasing $G$ beyond the saturation point \(G_{\rm sat}\triangleq \hat \beta \color{black}\lfloor\eta_{\rm Tx}\rfloor\color{black}=6\) does not improve DoF as it does not unlock new symmetric \(\beta\)-points. 
In contrast, the proposed scheduler exploits the additional Rx dimensions via user-dependent stream allocations, yielding more linearly decodable operating points and finer spatial-multiplexing control under the same target-user set. 
%

\color{black}
To quantify construction reliability and search effort in
Table~\ref{tab:greedy_success}, we use
$p_{\rm suc}\!=\!N_{\rm suc}/N_{\rm tr}$, where success requires a
complete regular C1--C2-certified schedule, and
$\overline N_{\rm FC}$ denotes the average number of feasibility
checks per trial. For these trials, at most $10$ candidate swaps are examined during repair.
The trials differ only through randomized ordering of admissible candidate and repair choices; hence, $p_{\rm suc}$ quantifies the empirical reliability of the greedy construction, whereas schedule existence is independently assessed
by an exhaustive regular-assignment check. 
Under the common search
setting in Table~\ref{tab:greedy_success}, all selected realizations
succeed in all $1000$ trials except
$\mathrm{DoF}_{\rm sch}\!=\!33$, for which
$p_{\rm suc}\!=\!0.995$. The construction effort is nonmonotonic,
and a higher DoF does not necessarily incur a larger subpacketization
factor: $\mathrm{DoF}_{\rm sch}\!=\!42$ is achieved with
$\delta\!=\!7$, versus $\delta\!=\!13$ at
$\mathrm{DoF}_{\rm sch}\!=\!39$.
\begin{table}[t]
\color{black}\centering
\caption{Greedy construction reliability for
$(L,G,\Omega,t)=(14,11,4,2)$ under
$(J_{\max},\tau,N_{\rm tr})=(20,2,1000)$.}
\vspace{-2mm}
\label{tab:greedy_success}
\scriptsize
\setlength{\tabcolsep}{2.8pt}
\renewcommand{\arraystretch}{1.05}
\begin{tabular}{@{}c c c c c c@{}}
\toprule
$\eta$ & $m$ & $\mathrm{DoF}_{\rm sch}$ &
$\delta$ & $p_{\rm suc}$ &
$\overline{N}_{\rm FC}$\\
\midrule
1 & 4 & 24 &  2 & 1.000 & 10.00  \\
2 & 2 & 30 &  5 & 1.000 & 34.17  \\
2 & 3 & 33 & 11 & 0.995 & 128.52 \\
3 & 1 & 39 & 13 & 1.000 & 78.00  \\
3 & 2 & 42 &  7 & 1.000 & 72.26  \\
\bottomrule
\end{tabular}
\end{table}
%
For Tx-side scaling at $(G,\Omega,t)=(11,4,2)$, varying
$L$ over $\{8,11,14\}$ yields
$\mathrm{DoF}_{\rm sch}=\{24,33,42\}$ and
$\overline{N}_{\rm FC}=\{10,128.52,72.26\}$, showing nonmonotonic
construction effort.
\color{black}
Fig.~\ref{fig:performance_comparison}(c) illustrates how finer DoF granularity translates into broader set of operating points for  enhanced finite-SNR symmetric rate optimization, compared with the symmetric scheme in~\cite{naseritehrani2026cache} and the MU-MIMO baseline with local caching gain only.
For fixed $\Omega=4$, the symmetric design has \(\CB_{\Omega=4} = \{3,6\}\), yielding $\mathrm{DoF}_{\mathrm{ref}}\in\{12, 24\}$, whereas the proposed scheduling expands the feasible stream allocations to $\beta_k\in\{4,5,7,8\}$, 
and hence the achievable DoF options to
$\mathrm{DoF}_{\mathrm{sch}} \in \{15, 18, 21, 27, 30\}$. 
The enlarged set of feasible stream allocations allows the scheduler to adaptively choose the rate-maximizing configuration at each SNR, resulting in better symmetric rates than the baselines 
in all SNR regions.

%

\section{Conclusion}
\label{sec:conclusions}
%
%
We characterized asymmetric stream allocation for one-shot linearly
decodable MIMO-CC delivery. The derived stream-count condition
generalizes the symmetric common-stream constraint to user-dependent
allocations, while the proposed greedy multicast scheduler fills the
DoF gaps left by symmetric scheduling and certifies linear decodability
of any returned schedule. Numerical results demonstrate finer DoF
granularity and finite-SNR symmetric-rate gains.

\bibliographystyle{IEEEtran}
\bibliography{conf_short,IEEEabrv,references,whitepaper}
\end{document}